\documentclass[10pt]{article}       % onecolumn (second format)
\usepackage{graphicx}
\usepackage{color,soul}
\usepackage{latexsym,amssymb,lastpage}
\usepackage{graphicx,amsfonts}
\usepackage{times,mathptmx,bm,amsmath,amsthm}
\usepackage{dcolumn}
\usepackage{algpseudocode}

\newcommand{\sM}{\mathcal{M}}
\newcommand{\sG}{\mathcal{G}}
\newcommand{\sV}{\mathcal{V}}
\newcommand{\sE}{\mathcal{E}}

\newcommand{\sB}{\mathcal{B}}
\newcommand{\hchi}{\widetilde{\chi}}
\newcommand{\tc}{\tilde{c}}

\newtheorem{thm}{Theorem}
\newtheorem{prop}[thm]{Proposition}
\newtheorem{lem}[thm]{Lemma}
\newtheorem{cor}[thm]{Corollary}

\begin{document}

\title{An $O(n \log n)$ time Algorithm for computing the Path-length Distance between Trees}

\author{David Bryant
\thanks{Department of Mathematics and Statistics, 
              University of Otago, Dunedin, New Zealand. {\tt david.bryant@otago.ac.nz}. {\tt https://orcid.org/0000-0003-1963-5535}}              
        \and
        Celine Scornavacca
        \thanks{ISEM, CNRS, Universit\'{e} de Montpellier, IRD, EPHE, 
	Montpellier, France {\tt  celine.scornavacca@umontpellier.fr}}
}


\maketitle

%\date{\today}

\begin{abstract}
Tree comparison metrics have proven to be an invaluable aide in the reconstruction and analysis of phylogenetic (evolutionary) trees. The path-length distance between trees is a particularly attractive measure as it reflects differences in tree shape as well as differences between branch lengths. The distance equals the sum, over all pairs of taxa, of the squared differences between the lengths of the unique path connecting them in each tree. We describe an $O(n \log n)$ time for computing this distance, making extensive use of tree decomposition techniques introduced by Brodal et al. \cite{Brodal04}.\\
{\bf keywords} Phylogeny; Tree comparison metrics;  Path-length metric; Tree decomposition.
\end{abstract}

\section{Introduction}
A \emph{phylogenetic tree} is a tree describing the evolution of a set of entities $X$ (species, genes etc.), which will be called \emph{taxa}  from now onwards.  Degree-one nodes are called \emph{leaves} and a bijective function associates each taxon to a leaf. Internal nodes represent putative ancestral taxa and branch lengths quantify the evolutionary distances between nodes.

Tree comparison metrics provide a  quantitative measure of the similarity or difference between two phylogenetic trees.  They have proven invaluable for statistical testing (e.g. \cite{Penny93,Holmes05,Susko11}), for visualisation \cite{Hillis05}, and for the construction of consensus trees \cite{Swofford91,Bryant03,Lapointe97}. By far the most well-known tree comparison metric is the Robinson-Foulds metric \cite{Robinson81}, which equals the number of bipartitions\footnote{A bipartition $A|B$ with $A\cup B=X$ is in a phylogenetic tree $T=(V,E)$ if there exists an edge $e\in E$ such that its removal creates two trees with taxon sets $A$ and $B$.}that are in one tree and not the other. However many other different metrics have also been proposed, each one based on a different characteristic of the trees being compared. 

Here we consider pairs of trees  on the same set of taxa. Also, our trees are \emph{binary}, i.e. each internal node has degree three.
The {\em path-length} between two taxa in a phylogenetic tree is the sum of the branch lengths along the unique path between them.  
The {\em path-length distance} between two trees $T_1$ and $T_2$  is given by
\begin{equation}
\Delta(T_1,T_2) = \sum_{ij} (p_{ij} - q_{ij})^2, \label{eq:Delta}
\end{equation}
where $p_{ij}$ is the path length between taxa $i$ and $j$ in the first tree and $q_{ij}$ is the path length in the second tree. We note that $\sqrt{\Delta(T_1,T_2)}$ is a metric in the mathematical sense. The first explicit description of the metric appears in \cite{Penny93} (without branch lengths) and  \cite{Lapointe97} (with branch lengths), though closely related ideas appear much earlier (e.g.  \cite{Hartigan67,Farris69,Williams71}). 

Given a phylogeny with $n$ leaves, it takes $O(n^2)$ time to construct the set of all path-lengths $p_{12},p_{13},\ldots,p_{(n-1)n}$, using the dynamic programming algorithm presented in \cite{Bryant97}. Hence the path-length distance can be easily computed in $O(n^2)$ time. Our main contribution in this paper is to show that we can compute this distance in $O(n \log n)$ time, which is almost, but not quite, linear in the size of the problem input.

Expanding \eqref{eq:Delta} gives 
\begin{equation}
\Delta(T_1,T_2) = \sum_{ij} (p_{ij})^2 + \sum_{ij} (q_{ij})^2 - 2 \sum_{ij} p_{ij} q_{ij}.\label{eq:DistExpand}
\end{equation}
The first two terms can be evaluated in linear time using dynamic programming, as outlined in   Section 2. To compute the second term efficiently we first introduce a tree decomposition technique (Section 3) allowing the sum to be evaluated in $O(n \log n)$ time (Section 4). Both the tree decomposition and algorithm of Section 4 draw heavily on an algorithm of \cite{Brodal04} for computing the quartet distance between two trees. 

\section{Sums of squared distances}

In this section we show how to compute the sum of squared distances $\sum_{ij} p_{ij}^2$ in a tree in linear time. We begin by introducing some notation that will be used in the rest of the paper.

Select an arbitrary leaf $\rho$ and consider both $T_1$ and $T_2$ as rooted trees with root $\rho$.  We think of $\rho$ being at the top of the tree and the other leaves being at the bottom of the tree.
For any two edges $e,e'$ we write $e \preceq e'$ if the path from $e$ to $\rho$ passes through $e'$.  We write $e \prec e'$ if $e \preceq e'$ and $e \neq e'$.  Hence if $e$ is the edge incident with the root $\rho$ then $e' \prec e$ for all other edges $e'$. We say that $e$ is external if it is incident to a leaf other than $\rho$; otherwise $e$ is internal. When $e$ is internal let $e_L$ and $e_R$ denote the edges incident and immediately below $e$. 

We will use $e,e'$ to denote edges in $T_1$ and $f,f'$ to denote edges in $T_2$.
We let $x_e$ denote the length of an edge $e$ in $T_1$ and $y_f$ the length of an edge $f$ in $T_2$. Let $A_{ij}$ denote the set of edges on the path from $i$ to $j$ in $T_1$ and let $B_{ij}$ denote the corresponding set in $T_2$. Hence
\[p_{ij} = \sum_{e \in A_{ij}} x_e \quad\quad q_{ij} = \sum_{f \in B_{ij}} y_f.\]
Let $n(e)$ denote the number of leaves $\ell$ such that the path from $\ell$ to $\rho$ passes through $e$. Define
\[\alpha(e) = \sum_{e' \preceq e} n(e') x_{e'}.\]

\begin{prop} \label{prop:singleTree}
\[ \sum_{ij} (p_{ij})^2  = \sum_{\mbox{ \footnotesize $e$ internal } } \Big[ x_e (n-n(e)) (2\alpha(e) - n(e)x_e)  + 2\alpha(e_L)\alpha(e_R) \Big] .\]
\end{prop}
\begin{proof}
Given two edges $e_1,e_2$ we let
\begin{equation}
\chi(e_1,e_2) = | \{ \mbox{pairs $ij$} : e_1,e_2 \in A_{ij}\} |,\label{eq:chi}\end{equation}
the number of pairs  having both $e_1$ and $e_2$ on the path between them. Then
\begin{align}
\sum_{ij} (p_{ij})^2 & = \sum_{ij} \left(\sum_{e_1 \in A_{ij}} x_{e_1} \right) \left( \sum_{e_2 \in A_{ij}} x_{e_2} \right) \nonumber \\
&= \sum_{ij} \sum_{e_1,e_2 \in A_{ij}}  x_{e_1} x_{e_2}\\
& = \sum_{e_1} \sum_{e_2} \sum_{ij: e_1, e_2 \in A_{ij}} x_{e_1} x_{e_2}\\
& = \sum_{e_1} \sum_{e_2} x_{e_1} x_{e_2} \chi(e_1,e_2)  \label{eq:pij2} 
\end{align}

If $e_1 \prec e_2$ then 
 $\chi(e_1,e_2) = n(e_1) (n - n(e_2))$.  Hence, for any $e_2$ we have
 \begin{align*}
 \sum_{e_1:e_1 \prec e_2} x_{e_1} x_{e_2} \chi(e_1,e_2) &=  x_{e_2} (n-n(e_2)) 
 %\sum_{e_1 \prec e_2} x_{e_2} n(e_2) \\
\sum_{e_1:e_1 \prec e_2} x_{e_1} n(e_1) \\
 & =x_{e_2} (n-n(e_2)) (\alpha(e_2)-n(e_2)x_{e_2}).
 \end{align*}
 
 If $e_1 \not \preceq e_2$ and $e_2 \not \preceq e_1$  then  $\chi(e_1,e_2) = n(e_1) n(e_2)$. Furthermore there is an edge $e$ with children $e_L,e_R$ such that, without loss of generality, $e_1 \preceq e_L$ and $e_2 \preceq e_R$.  For such an edge $e$ we have
 \begin{align*}
 \sum_{e_1:e_1 \preceq e_L} \sum_{e_2:e_2 \preceq e_R} x_{e_1} x_{e_2} \chi(e_1,e_2) & = 
\sum_{e_1:e_1 \preceq e_L} \sum_{e_2:e_2 \preceq e_R} x_{e_1} x_{e_2} n(e_1)n(e_2) \\
& = \alpha(e_L) \alpha(e_R).
\end{align*}

Summing up over all $e_1,e_2$ in \eqref{eq:pij2} we have
\begin{align*}
\sum_{ij} (p_{ij})^2 & = \sum_{e_1} \sum_{e_2} x_{e_1} x_{e_2} \chi(e_1,e_2)  \\
& = 2\sum_{e_2} \sum_{e_1 \prec e_2} x_{e_1} x_{e_2} \chi(e_1,e_2)  + 
2 \sum_e \sum_{e_1 \preceq e_L} \sum_{e_2 \preceq e_R} x_{e_1} x_{e_2} \chi(e_1,e_2)  
+ \sum_e x_{e} x_{e} \chi(e,e) \\
& = 2 \sum_{e_2}  x_{e_2} (n-n(e_2)) (\alpha(e_2)-n(e_2)x_{e_2}) +  2 \sum_e \alpha(e_L) \alpha(e_R) + \sum_e x_{e} x_{e} n(e)(n-n(e))
\end{align*}
and the result follows.  \end{proof}

\begin{prop} \label{prop:sumSquared}
The sum $\sum_{ij} (p_{ij})^2$ can be computed in linear time.
\end{prop}
\begin{proof}
If $e$ is external, $n(e) = 1$ and $\alpha(e) = x_e$. Otherwise 
\begin{align*}
n(e) & = n(e_L) + n(e_R) \\
\alpha(e) &= \alpha(e_L) + \alpha(e_R) + n(e) x_e.
\end{align*}
Hence with a post-order traversal of the tree we can compute $n(e)$ and $\alpha(e)$ for all edges $e$ in $O(n)$ time.  Computing the sum takes a further $O(n)$ time by Proposition~\ref{prop:singleTree}. $ \sum_{ij} (q_{ij})^2$ can be computed in the same way.
  \end{proof}

\section{Segment decomposition}

In this section we introduce a hierarchical decomposition of the edge set of $T_2$ that forms the structure used in our dynamical programming algorithm in Section 4.

Let $Q$ be a connected subset of $E(T_2)$, the set of edges of $T_2$. We define the {\em boundary} of $Q$ to be the set of vertices incident both to edges within $Q$ and to edges outside $Q$:
\[ \partial Q = \{v : \mbox{ there are $e \in Q$, $e' \not \in Q$ incident with $v$ }\}.\]
The {\em degree} of $Q$ is the cardinality of $\partial Q$.  A {\em segment} of $T_2$ is a connected subset of $E(T_2)$ with degree at most two.

A {\em segment decomposition} for $T_2$ is a binary tree $T_D$ such that
\begin{enumerate}
\item[(D1)] The leaves of $T_D$ correspond to edges in $E(T_2)$ (i.e. minimal segments);
\item[(D2)] Each node of $T_D$ corresponds to a segment of $T_2$;
\item[(D3)] The segment corresponding to an internal node of $T_D$ equals the disjoint union of the segments corresponding to its children.
\end{enumerate}
An example of segment decomposition is given in Figure \ref{segDec}. 

The main result in this section is that we can, in linear time, construct a segment decomposition for $T$ with height $O( \log n)$. 

\begin{figure}[htbp]
\begin{center}
\includegraphics[width=0.8\linewidth]{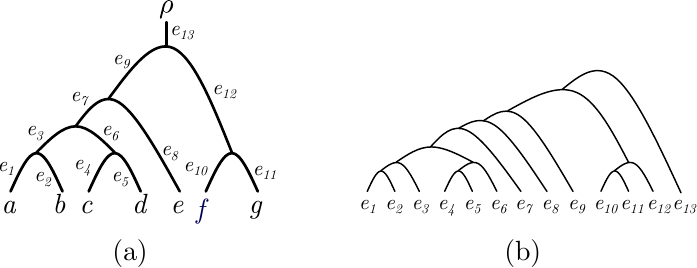}
\caption{(a) A phylogenetic tree and (b) a segment decomposition for it.}
\label{segDec}
\end{center}
\end{figure}

The definition of a segment decomposition is based on the tree decomposition used by \cite{Brodal04} to compute quartet-based distances, which in turn are based on techniques for efficient parsing of expressions \cite{Brent74,Cohen95}. The main difference with  \cite{Brodal04}  is that the segment decomposition is based on partitioning the set of edges, rather that the set of vertices, and that we were able to obtain a tighter bound on the height.

Our algorithm for constructing $T_D$ is agglomerative: we start with a  degree one vertex for each edge in $E(T_2)$; these form the leaves of $T_D$. Each iteration, we identify pairs of maximal nodes corresponding to pairs of segments which can be combined to give new segments. We make the nodes in each pair children of a new node. The process continues until one node remains and $T_D$ is complete. 

The following Proposition shows that in any partition of $E(T_2)$ into segments we can always find a large number of pairs of disjoint segments which can be merged to give other segments. 

\begin{prop}~\label{manypairs}  
Let $T$ be a binary tree. Let $\sM$ be a collection of segments which partition $E(T)$. Then there are at least $\frac{|\sM|}{4}$ non-overlapping pairs $(A,B)$ such that $A,B \in \sM$ and $A \cup B$ is a segment of $T$.
\end{prop}
\begin{proof}
Let $\sG_\sM = (\sV_\sM,\sE_\sM)$ be the graph with vertex set  
\[ \sV_\sM = \bigcup_{A \in \sM} \partial A\]
and edge set 
\[ \sE_\sM = \left \{ \{u,v\}:  \partial A = \{u,v\} \mbox{ for some $A \in \sM$ } \right\}.\]
Decompose $\sG_\sM$ into maximal paths $P_1,P_2,\ldots,P_\kappa$ which contain no degree three vertices in their interiors. 
For each $i$, let $\sM_i$ be the set of elements  $A \in \sM$ such that $\partial A \subseteq P_i$. The sets $\sM_i$ partition $\sM$.

Fix one path $P_i = v_1,v_2.\ldots,v_\ell$. We order the elements of $\sM_i$ lexicographically with respect to the indices of their boundary vertices. In other words, if $A,B \in \sM_i$ satisfy $\partial A = \{v_j,v_k\}$ and $\partial B = \{v_\ell,v_m\}$ (where we might have $j=k$ or $\ell = m$) then we write $A < B$ if  $\max(j,k) < \max(\ell,m)$ or $\max(j,k) = \max(\ell,m)$ and $\min(j,k)<\min(\ell,m)$. With this ordering, if  $A_k$ and $A_{k+1}$ are adjacent then $(A_k \cup A_{k+1})$ is connected and has degree at  most two.  Hence by pairing off $A_1$ and $A_2$, $A_3$ and $A_4$, and so on, we can construct $\lfloor \frac{\sM_i}{2} \rfloor$ disjoint pairs.  An example is given in Figure \ref{paring}.

\begin{figure}[htbp]
\begin{center}
\includegraphics[width=0.7\linewidth]{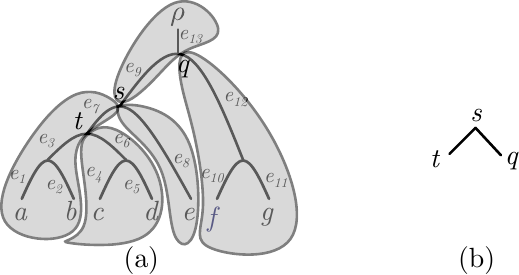}
\caption{(a) A phylogenetic tree with the segment decomposition  $\sM=\{\{e_1,e_2,e_3,e_7\},\{e_4,e_5,e_6\},\{e_8\},\{e_9,e_{13}\},\{e_{10},e_{11},e_{12}\}\}$ drawn on  it, with boundary sets respectively $\{t,s\},\{t,t\},\{s,s\},\{s,q\},\{q,q\}$. (b) The corresponding $\sG_\sM$. For this decomposition, there is a single maximal path $P_1 =t,s,q :=v_1,v_2,v_3$ and boundary sets become respectively $\{v_1,v_2\},\{v_1,v_1\},\{v_2,v_2\},\{v_2,v_3\},\{v_3,v_3\}$. Thus, the ordering of $\sM$ is $\{\{e_4,e_5,e_6\},\{e_1,e_2,e_3,e_7\},\{e_8\},\{e_9,e_{13}\},\{e_{10},e_{11},e_{12}\}\} $. }
\label{paring}
\end{center}
\end{figure}

The total number of pairs we obtain this way is given by $\sum_{i=1}^\kappa \lfloor \frac{|\sM_i|}{2} \rfloor$.  We will determine a lower bound for this sum. Let $d$ be the number of degree three vertices in $\sG_\sM$. \ Since $\sG_\sM$ is connected and acyclic there are $d+2$ paths $P_i$ which contain a degree one vertex in $\sG_\sM$ and $d-1$ paths which do not. If $P_i$ contains a degree one vertex then $\sM_i$ contains at least one component with degree two and another component with boundary equal to the degree one vertex, so $|\sM_i| \geq 2$. 
If $P_i$ contains no degree one vertices then $|\sM_i|$ is at least one. Let $x$ denote the number of paths $P_i$ which contain a degree one vertex and for which $|\sM_i|$ is odd (and hence at least three). We have
\[|\sM| = \sum_{i=1}^\kappa |\sM_i| \geq  3x + 2(d+2 - x) + (d-1) = x + 3d + 3\]
as well as $0 \leq x \leq d+2$ and $d \geq 0$. 

We have that $|\sM_i|$ is even for at least $(d+2) - x$ paths ending in a degree one vertex, and for these paths $\frac{|\sM_i|}{2} = \lfloor \frac{|\sM_i|}{2} \rfloor$. Thus
\[\frac{|\sM|}{2} - \sum_{i=1}^\kappa \lfloor \frac{|\sM_i|}{2} \rfloor  \leq \frac{x}{2} + \frac{d-1}{2}.\]
To bound the right hand side, note that the linear program
\begin{align*}
\max \quad & x+d \\
\mbox{subj. to} \quad & x -d \leq 2 \\
& x + 3d \leq |\sM| - 3
\end{align*}
has solution $d = \frac{|\sM|-5}{4}$, $x = \frac{|\sM|+3}{4}$ and so $x+d \leq \frac{2|\sM| - 2}{4}$. Hence
\[ \frac{|\sM|}{2} - \sum_{i=1}^\kappa \lfloor \frac{|\sM_i|}{2} \rfloor  \leq  \frac{|\sM|}{4} - \frac{3}{4} \]
and $\sum_{i=1}^\kappa \lfloor \frac{|\sM_i|}{2} \rfloor$, the number of pairs, is bounded below by  $\frac{|\sM|}{4}$.
 \end{proof}

We can now state the algorithm for constructing $T_D$. Initially $T_D$ is a set of isolated vertices. As the algorithm progresses, vertices are combined into larger trees, so that each iteration $T_D$ is a forest. The algorithm terminates when $T_D$ contains a single tree. 

At each iteration let $\sM$ denote the partition of the edge set of $E(T)$ into segments corresponding to the maximal elements of the incomplete tree $T_D$. Rather than store this partition explicitly, we maintain a linked list $\sB$ of boundary nodes. For each element $v$ in the list we maintain pointers to maximal nodes $T_D$ corresponding to segments in $\sM$ having $v$ in their boundaries. In addition, we maintain pointers from each node in $T_D$ to the boundary nodes of the corresponding segments.  

\begin{enumerate}
\item Initialize $T_D$ with a forest of degree-one vertices corresponding to each edge of $E(T_2)$. Hence we initialise $\sB$ with one element for each vertex in $V(T_2)$, with the associated pointers. At this point, $\sM$ is the partition of $E(T_2)$ putting each edge into a separate block.
\item {\bf While} $T_D$ is disconnected {\bf do} 
\begin{enumerate}
\item Using the construction in Proposition~\ref{manypairs} determine a set of at least $k \geq \frac{|\sM|}{4}$ pairs $(A_1,B_1),\ldots,(A_k,B_k)$ of disjoint elements of $\sM$ such that $A_j \cup B_j$ has at most two boundary points. \label{lemstep}
\item For each pair $(A_i,B_i)$, $i=1,2,\ldots,k$, create a new node of $T_D$ corresponding to $A_i \cup B_i$ and with children corresponding to $A_i$ and $B_i$.
\item Update the list $\sB$ of boundary vertices and the associated pointers.
\end{enumerate}
\end{enumerate}

\begin{thm}
We can construct a segment decomposition tree $T_D$ for $T_2$ with height $O(\log n)$ in $O(n)$ time.
\end{thm}
\begin{proof}
We only merge nodes if the union of their corresponding segments is also a segment. Hence $T_D$ will be a segment decomposition tree. It remains to prove the bound on height and running time.

We note that $|\sM|$ reduces by a factor of $\frac{3}{4}$ each iteration. Hence the number of iterations is at most $ \log_{\frac{4}{3}} (2n-3)$, which is also a bound on the height of the tree. 

Using the list of boundary points $\sB$ we can construct construct $\sG_{\sM}$ and identify pairs, in $O(|\sM|)$ time each iteration. Thus the total running time is  at most $O( n (\sum_{i=0} \left( \frac{3}{4} \right)^i )) = O(n)$ time. 
 \end{proof}

We can strengthen the height bound. We say that a tree is $k$-locally balanced if, for all nodes $v$ in
the tree, the height of the subtree rooted at $v$ is at most $k \cdot (1 + log|v|)$.  As the algorithm can be applied recursively on each node of $T_D$ we have that the global height bound applies to each node. Hence

\begin{cor}
The segment decomposition $T_D$ is ($1/log \frac{4}{3}$)-
locally balanced.
\end{cor}

\section{Computing the inner product}

 In this section we show that $\sum_{ij} p_{ij} q_{ij}$ can be computed in $O(n \log n)$ time, so that the main result follows from Eq.~\eqref{eq:DistExpand}.

A \emph{(taxon) colouring} is an assignment $c$ of the colors black and white to the taxa. For each edge $e$ of $T_1$ we let $c_e$ denote a coloring assigning black to those taxa on one side of $e$ and white to those on the other. For each edge $f$  in $E(T_2)$ and each colouring $c$ of the set of taxa, we let $\hchi(c,f)$ denote the number of pairs of taxa $ij$ such that  $i$ and $j$ have different colours and they label leaves on different sides of $f$.

\begin{lem} \label{lem:InnerIsSum}
\begin{equation}
\sum_{ij} p_{ij} q_{ij}  =  \sum_{e \in E(T_1) } \sum_{f \in E(T_2) } x_{e} y_{f} \hchi(c_e,f)  \label{eq:pijqij} 
\end{equation}
\end{lem}
\begin{proof}
\begin{align}
\sum_{ij} p_{ij} q_{ij} & = \sum_{ij}  \left(\sum_{e:e \in A_{ij}} x_{e} \right) \left( \sum_{f:f \in B_{ij}} y_{f} \right) \nonumber \\
&= \sum_{ij} \sum_{e \in A_{ij}}   \sum_{f \in B_{ij}} x_e y_{f}\\
& = \sum_{e \in E(T_1)} \sum_{f \in E(T_2)} x_{e} y_{f} \hchi(c_e,f)  .
\end{align}
 \end{proof}

For the remainder of this section we will assume that the vertices in $T_2$ are indexed $v_1,v_2,\ldots,v_{2n-3}$. The actual ordering does not matter; it is only used to help presentation.

Let $T_D$ be the segment decomposition tree constructed for $T_2$ using the Algorithm in Section 3. For each node $v$ of $T_D$ we let $Q_v \subseteq E(T_2)$ denote corresponding segment in $T_2$.
The overall strategy at this point is to compute values for each node in $T_D$ which will allow us to: (i) compute, for an initial choice of  $e \in E(T_1)$, the sum $ \sum_{f \in E(T_2) } x_{e} y_{f} \hchi(c_e,f) $ in linear time, and 
(ii) update this computation efficiently as we iterate in a particular way through edges $e$ of $T_1$.

We will store three pieces of information at every non-root node $v$ of $T_D$, the exact type of information stored being dependent on the degree of the segment $Q_v$ corresponding to $v$.\\
If $Q_v$ is degree one then we store: 
\begin{itemize}
\item[$\circ$] Two integer counts $w_v,b_v$
\item[$\circ$]  A description (e.g. coefficients) for a quadratic polynomial $\phi_v(\cdot,\cdot)$ with two variables. %\todo{Since we know the number of leaves on either side, we really only need to store one variable in the first case and two in the order. However I'm not sure whether there is enough of an advantage to using fewer variables here}
\end{itemize} 
If $Q_v$ has degree two then we store:
\begin{itemize}
\item[$\circ$]  Two integer counts $w_v,b_v$
\item[$\circ$] A description (e.g. coefficients) for a quadratic polynomial $\phi_v(\cdot,\cdot,\cdot,\cdot)$ with four variables.
\end{itemize}

We now show how the values $b_v,w_v$ and $\phi_v$ are computed using a colouring $c$ of the taxa. We start at the leaves of $T_D$ and work upwards towards the root.

Suppose that $v$ is a leaf of $T_D$, so that $Q_v$ contains a single edge $f$ of $T_2$. There are two cases. 
\begin{enumerate}
\item The edge  $f$ is  incident with a leaf $u$ of $T_2$, so $Q_v$ has degree one. If $c(u)$ is black then $b_v = 1$ and $w_v = 0$, while if $c(u)$ is white we have $w_v = 1$ and $b_v = 0$. In either case 
\begin{equation} \phi_v(b,w) = y_f (b \cdot w_v + w \cdot b_v).\end{equation} 
\item The edge $f$ is not incident with a leaf of $T_2$, so $Q_v$ has degree two. Then $b_v = w_v = 0$ and 
\begin{equation} \phi_v(b_1,w_1,b_2,w_2) = (b_1 w_2 + b_2 w_1) y_f.\end{equation} 
\end{enumerate}

Now suppose that $v$ is an internal vertex of $T_D$. % Let $Q_{v_L}$ denote the segment of $T_2$ corresponding to the left child $v_L$ of $v$ and let $Q_{v_R}$ be the segment corresponding to the right child.  %Let $b_{v_L},w_{v_L},b_{v_R},w_{v_R}$ denote the counts $b_{v_L}$, $w_{v_L}$, $b_{v_R}$, and $w_{v_R}$. %Let $\phi_L$ and $\phi_R$ denote the polynomials $\phi_{v_L}$ and $\phi_{v_R}$. \todo{Not sure if we need to simplify the subsubscripts like this. Do you think it helps?}
Once again there are several cases, however in all cases we have 
\begin{align*}
b_v & = b_{v_L} + b_{v_R} \\
w_v & = w_{v_L} + w_{v_R}.
\end{align*}
\begin{enumerate}
  \setcounter{enumi}{2}
\item Suppose $Q_{v_L}$ and $Q_{v_R}$ have degree one. Then
\begin{equation} \phi_v(b,w) = \phi_{v_L}(b+b_{v_R},w+w_{v_R}) +  \phi_{v_R}(b+b_{v_L},w+w_{v_L}).  \label{eq:Qv1} \end{equation}
\item Suppose $Q_{v_L}$ has degree two and $Q_{v_R}$ has degree one, where $\partial Q_{v_L} = \{v_i,v_j\}$ and $Q_{v_R} = \{v_j\}$. 
\begin{align}
\intertext{(a) If $Q_v$ has degree one and $i<j$ then}
\phi_v(b,w) &= \phi_{v_L}(b,w,b_{v_R},w_{v_R}) + \phi_{v_R}(b+b_{v_L},w+w_{v_L});\\
\intertext{(b) If $Q_v$ has degree one and $i>j$ then}
\phi_v(b,w) &= \phi_{v_L}(b_{v_R},w_{v_R},b,w) + \phi_{v_R}(b+b_{v_L},w+w_{v_L});\\
\intertext{(c) If $Q_v$ has degree two and $i<j$ then} 
\phi_v(b_1,w_1,b_2,w_2) &= \phi_{v_L}(b_1,w_1,b_2+b_{v_R},w_2+w_{v_R}) + \phi_{v_R}(b_1+b_2+b_{v_L},w_1+w_2+w_{v_L});\\
\intertext{(d) If $Q_v$ has degree two and $i>j$ then}
\phi_v(b_1,w_1,b_2,w_2) &= \phi_{v_L}(b_1+b_{v_R},w_1+w_{v_R},b_2,w_2) + \phi_{v_R}(b_1+b_2+b_{v_L},w_1+w_2+w_{v_L}).
\end{align}
\item The case when $Q_{v_L}$ has degree one and $Q_{v_R}$ has degree two is symmetric.
\item Suppose that $Q_{v_L}$ and $Q_{v_R}$ have degree two, that $\partial Q_{v_L} = \{v_i,v_j\}$ and $\partial Q_{v_R} = \{v_j,v_k\}$. We can assume that $i<k$ since the alternative case follows by  symmetry. This leaves three possibilities:
\begin{align}
\intertext{(a) If $i<j$ and $j<k$ then} 
\phi_v(b_1,w_1,b_2,w_2) &= \phi_{v_L}(b_1,w_1,b_2+b_{v_R},w_2+w_{v_R}) + \phi_{v_R}(b_1+b_{v_R},w_1+w_{v_R},b_2,w_2);\\
\intertext{(b) If $i<j$ and $j>k$ then} 
%\intertext{If $i<j$ and $k>j$ then} 
\phi_v(b_1,w_1,b_2,w_2) &= \phi_{v_L}(b_1,w_1,b_2+b_{v_R},w_2+w_{v_R}) + \phi_{v_R}(b_2,w_2,b_1+b_{v_L},w_1+w_{v_L});\\
\intertext{(c) If $j<i$ and (hence) $j<k$ then}
\phi_v(b_1,w_1,b_2,w_2) &= \phi_{v_L}(b_1+b_{v_R},w_1+w_{v_R},b_1,w_1) + \phi_{v_R}(b_1+b_{v_R},w_1+w_{v_R},b_2,w_2). \label{eq:Qvend}
\end{align}
\end{enumerate}

An illustration for several of these cases can be found in Figure \ref{cases} below.
\begin{figure}[htbp]
\begin{center}
\includegraphics[width=0.7\linewidth]{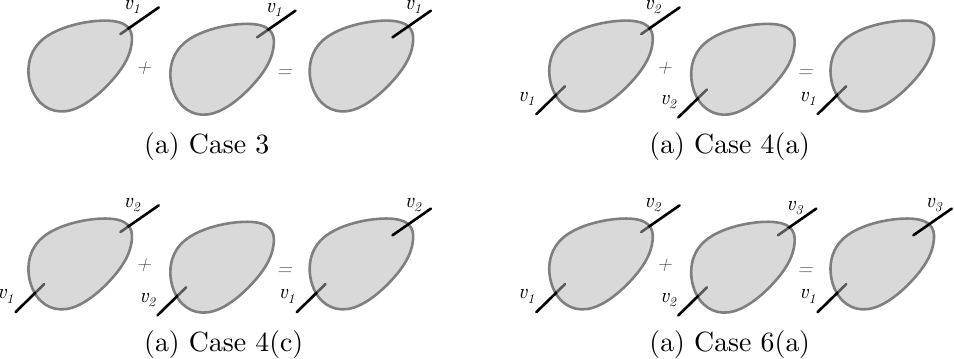}
\caption{Cartoons of segment merging for several cases discussed in the main text.}
\label{cases}
\end{center}
\end{figure}

\begin{lem}
Suppose that $b_v,w_v$ and $\phi_v$ have been computed as above for all nodes of $T_D$ except the root. Let $v_L$ and $v_R$ be the children of the root of $T_D$. Then 
\[\sum_{f \in E(T_2)} \hchi(c,f) y_f = \phi_{V_L}(b_{v_R},w_{v_R}) + \phi_{V_R}(b_{v_L},w_{v_L}).\]
\end{lem}
\begin{proof}
For any node $v$ of $T_D$  we let $L_v$ denote the set of leaves of $T_2$ not incident with  an edge of $Q_v$.  If $Q_v$ has degree two and boundary $\{v_i,v_j\}$, $i<j$, then we let $L^{(1)}_v$ be the leaves in $L_v$ which are closest to $v_i$ and $L^{(2)}_v$ the leaves in $L_v$ which are closest to $v_j$. Let $\tc$ be any colouring of the leaves of $T_2$, possibly distinct from $c$. Let $B$ and $W$ be the sets of leaves that $\tc$ colours black and white respectively. 

We will establish the following claims for all nodes $v$ in $T_D$, using induction on the height of the node.
\begin{enumerate}
\item[(C1)] $b_v$ and $w_v$ are the  number of leaves incident with edges in $Q_v$ which are coloured black and white by $c$ (and hence by $\tc$).
\item[(C2)] If $Q_v$ has degree one, $b = |B \cap L_v|$ and $w = |W \cap L_v|$ then
\[ \sum_{f \in Q_v} \hchi(\tc,f) y_f = \phi_v( b,w ).\]
\item[(C3)] If $Q_v$ has degree two, $b_1 = |B \cap L^{(1)}_v|$, $w_1 = |W \cap L^{(1)}_v|$, $b_2 = |B \cap L^{(2)}_v|$, and $w_2 = |B \cap L^{(2)}_v|$ then
\[ \sum_{f \in Q_v} \hchi(\tc,f) y_f = \phi_v( b_1,w_1,b_2,w_2 ).\]
\end{enumerate}

We start by considering  any leaf $v$ of $T_D$. In this case,  $Q_v$ contains a single edge $f$. If $f$ is an edge incident to a leaf coloured white then $b_v = 0$, $w_v = 1$ as required, and 
$\hchi(\tc,f)$ equals the number of leaves coloured black by $\tc$, so
\[  \hchi(\tc,f) y_f = |B \cap L_v| y_f = (b w_v + w b_v)y_f = \phi_v(b,w).\]
The same holds if the leaf is coloured black.\\
If the edge $f$ is internal then $b_v=w_v=0$, and $\hchi(\tc,f)$ is equal to the number of paths crossing $f$ connecting leaves with different colours, or 
\[ |B \cap L_v^{(1)}| |W \cap L_v^{(2)}| +  |W \cap L_v^{(1)}| |B \cap L_v^{(2)}| =
b_1 w_2 + w_1 b_2,\]
so $\hchi(\tc,f) y_f = \phi_v(b_1,w_1,b_2,w_2)$.

Now consider the case when $v$ is an internal node of $T_D$, other than the root.  Let $v_L$ and $v_R$ be the two children of $v$. %As above, we use $b_{v_L},w_{v_L},Q_{v_L}$ in place of $b_{v_L}$, $w_{v_L}$ and $Q_{v_L}$; likewise for $b_{v_R},w_{v_R}$ and $Q_{v_R}$. 
Note that $Q_v$ is the disjoint union of $Q_{v_L}$ and $Q_{v_R}$, so $b_v = b_{v_L}+b_{v_R}$ and $w_v = w_{v_L}+w_{v_R}$, proving (C1).  \\
Furthermore, we have
\[\sum_{f \in Q_v} \hchi(\tc,f)  = \sum_{f \in Q_{v_L}} \hchi(\tc,f) + \sum_{f \in Q_{v_R}} \hchi(\tc,f).\]
If $Q_{v_L}$ has degree one then, by the induction hypothesis,
\[\sum_{f \in Q_{v_L}} \hchi(\tc,f)  = \phi_{v_L}(b',w')\]
where $b'$ and $w'$ are the numbers of leaves coloured black and white that are not incident with edges in $Q_{v_L}$. Similarly, if $Q_{v_L}$ has degree two then, by the induction hypothesis, 
\[\sum_{f \in Q_{v_L}} \hchi(\tc,f)  = \phi_{v_L}(b_1',w_1',b_2',w_2')\]
where $b_1'$ and $w_1'$ are the numbers of leaves coloured black and white that are not incident with edges in $Q_{v_L}$ and are closer to the boundary vertex of $Q_{v_L}$ with the smallest index, while 
$b_2'$ and $w_2'$ are the numbers of leaves coloured black and white that are not incident with edges in $Q_{v_L}$ and are closer to the boundary vertex of $Q_{v_L}$ with the largest index. The symmetric result holds for $Q_{v_R}$. \\
The different cases in Eq. \eqref{eq:Qv1} to Eq. \eqref{eq:Qvend} now correspond to the different counts for $b',w'$ or for $b_1',w_1',b_2',w_2'$ depending on whether $Q_{v_L}$ and $Q_{v_R}$ have degree one or two, and whether the boundary vertices in common had the highest or lower index for each segment. 

Now suppose that $v_L$ and $v_R$ are the children of the root of $T_B$. Then $\partial (Q_{v_L} \cup Q_{v_R}) = \emptyset$ so $Q_{v_L}$ and $Q_{v_R}$ must both have degree one. We have that $E(T_2)$ is the disjoint union of $Q_{v_L}$ and $Q_{v_R}$. Any leaf not incident to a leaf in $Q_{v_L}$ is incident to a leaf in $Q_{v_R}$ and vice versa. Hence
%\begin{align*}
%\sum_{f \in E(T_2)} \hchi(\tc,f) & = \sum_{f \in Q_{v_L}} \hchi(\tc,f)  +  \sum_{f \in Q_{v_R}} \hchi(\tc,f) \\
%& = \phi_{v_L}(b_{v_R},w_{v_R}) + \phi_{v_R}(b_{v_L},w_{v_L})
%\end{align*}
as required.
 \end{proof}

Evaluating Eq. \eqref{eq:Qv1} to Eq. \eqref{eq:Qvend} takes constant time and space per each node of $T_D$, since we manipulate and store a constant number of polynomials with at most four variables and total degree at most two. Thus, 
evaluating Eq. \eqref{eq:Qv1} to Eq. \eqref{eq:Qvend}  takes $O(n)$ time and space for each colouring, and since we want to sum this quantity over all colourings $c_e$ from edges $e \in E(T_1)$ a naive implementation would still take $O(n^2)$ time. The key to improving this bound is in the use of efficient updates. 

\begin{lem}
Suppose that we have computed $b_v$, $w_v$ and the functions $\phi_v$ for all $v \in T_D$, using a leaf colouring $c$. Let $\tc$ be a colouring which differs from $c$ at  $k$ leaves. Then we can update the values $b_v$, $w_v$ and the functions $\phi_v$ in $O(k + k \log(n/k)$ time. 
\end{lem}
\begin{proof}
Let $F'$ be the set of edges of $T_2$ which are incident to a leaf for which $c$ and $\tc$ have a different colour, so $|F'| = k$.  The only nodes $v$ in $T_D$ which need to updated are those with $f \in Q_v$ for some $f \in F'$. This is a union of the paths from $k$ leaves of $T_D$ to the root of $T_D$, and   so by Lemma 2  of \cite{Brodal04} , it has size $O(k+k \cdot log(\frac{n}{k}))$.
 \end{proof}

The final step is to show that we can navigate the edges in $E(T_1)$ so that the total number of changes in the colourings is bounded appropriately.  Suppose that $T_1$ is rooted at the leaf $\rho$ (the same as $T_2$). For each internal node $u$ in $T_1$ we let Small($u$) denote the child of $u$ with the smallest number of leaf descendants and let Large($u$) denote the child with the largest number of leaf descendants, breaking ties arbitrarily. 

The following recursive procedure returns the sum of 
\[\sum_{f \in E(T_2) } x_{e} y_{f} \hchi(c_e,f)  \]
over all edges $e \in E(T_1)$. Initially we let $e$ be the edge incident with the root $\rho$. Let $c$ be the colouring where  $\rho$  is black and all other leaves white. We initialise $T_D$ and fill out the values $b_v$, $w_v$ and $\phi_v$ for all nodes $v$ of $T_D$ using the colouring $c$. We then call {\sc Sum}($u$) where $u$ is the unique internal node adjacent to $\rho$.

\begin{figure}[htp]
\begin{algorithmic}
\Procedure{\sc Sum}{$u$}
\small{
\State Let $e$ be the edge connecting $u$ to its parent (in $T_1$).
\State $x \gets \displaystyle\sum_{f \in E(T_2) } x_{e} y_{f} \hchi(c,f)$, computed using $T_D$.
\If{ $u$ is a leaf} 
\State Color $u$  black and update $T_D$
\State return $x$
\Else 
\State Color the leaves in the subtree of $T_1$ rooted at Small($u$) in black and update $T_D$
\State $y \gets$ {\sc Sum}(Large($u$))
\State Color the leaves in the subtree of $T_1$ rooted at Small($u$) in white and update $T_D$
\State $z \gets$ {\sc Sum}(Small($u$))
\State return $x$ + $y$ + $z$
\EndIf
}
\EndProcedure
\end{algorithmic}
\begin{center}
\caption{Recursive algorithm {\sc Sum}}
\end{center}
\end{figure}

We see that the algorithm makes a pre-order traversal of $T_1$, evaluating the sum 
\[\sum_{f \in E(T_2) } x_{e} y_{f} \hchi(c_e,f)\]
for each edge $e$ and accumulating the total. Thus by Lemma~\ref{lem:InnerIsSum}, the algorithm returns $\sum_{ij} p_{ij} q_{ij}$.

The running time is dominated by the time required to update $T_D$. For each leaf, the update is made after only one leaf changes colour, so this takes $O(n \log n)$ summed over all leaves. For every other node $u$ in the tree, the number of nodes of $T_D$ to update is $O(k + k \log(n/k))$ where $k$ is the number of leaves in the subtree rooted at Small($u$). 

\begin{lem}
Let $T$ be a rooted binary tree with $n$ leaves and for each internal node $u$ of $T$ let $k_u$ denote the number of leaves in the smallest subtree rooted at a child of $u$. Then 
\[\sum_{\mbox{\footnotesize $u$ internal}} k_u \log(n/k_u) \leq n \log n.\]
\end{lem}
\begin{proof}
This is a restatement of Lemma 7 in \cite{Brodal04}.
 \end{proof}

\begin{thm}
Algorithm {\sc Sum} computes $\sum_{ij} p_{ij} q_{ij}$ in $O(n \log n)$ time. Hence the path length distance between $T_1$ and $T_2$ can be computed in $O(n \log n)$ time.
\end{thm} 

\subsection*{Acknowledgements}
This research was made possible due to travel funds made available from a Marsden grant to DB.

\bibliographystyle{spmpsci}      % mathematics and physical sciences
\bibliography{pathlength}   % name your BibTeX data base

\end{document}